\documentclass[11pt]{article}
\usepackage{amssymb,amsmath,amsthm,mathrsfs}
\usepackage{epsfig}
\usepackage[all]{xy}
\usepackage[labelfont=bf, margin=1cm]{caption}
\usepackage[backref=page]{hyperref}
\usepackage{footnote}

\usepackage[left=1.25in,right=1.25in,bottom=1.25in,top=1.25in]{geometry}

\newcommand{\MaxNAND}{\textsc{Max NAND}}
\newcommand{\MinDist}{\textsc{Min Dist}}
\newcommand{\NCP}{\textsc{NCP}}

\newcommand{\F}{\mathbb{F}}

\newcommand{\tensor}{\otimes}

\newcommand{\etal}{{et al.}}

\theoremstyle{plain}
\newtheorem{theorem}{Theorem}[section]
\newtheorem{lemma}[theorem]{Lemma}
\newtheorem{corollary}[theorem]{Corollary}

\newtheorem{claim}[theorem]{Claim}

\theoremstyle{definition}
\newtheorem{definition}[theorem]{Definition}
\newtheorem{fact}[theorem]{Fact}
\newtheorem{remark}[theorem]{Remark}

\DeclareMathOperator{\Opt}{{\sf Opt}}
\newcommand{\C}{\mathcal{C}}

\begin{document}

\title{A Simple Deterministic Reduction for \\ the Gap Minimum Distance of Code Problem}
\author{Per Austrin\thanks{Research done while at New York University supported by NSF Expeditions grant CCF-0832795.}\\
  University of Toronto
  \and
  Subhash Khot\thanks{Research supported by NSF CAREER grant CCF-0833228, NSF Expeditions grant CCF-0832795,
and BSF grant 2008059.}  \\
  New York University 
}

\maketitle

\begin{abstract}
We present a simple deterministic gap-preserving reduction from SAT to
the Minimum Distance of Code Problem over $\F_2$.  We also show how to
extend the reduction to work over any finite field.  Previously a
randomized reduction was known due to Dumer, Micciancio, and Sudan
\cite{DMS}, which was recently derandomized by Cheng and Wan
\cite{CW1, CW2}. These reductions rely on highly non-trivial coding
theoretic constructions whereas our reduction is {\it elementary}.

As an additional feature, our reduction gives a constant factor
hardness even for asymptotically good codes, i.e., having constant
rate and relative distance.  Previously it was not known how to
achieve deterministic reductions for such codes.
\end{abstract}

\pagebreak

\section{Introduction}
The Minimum Distance of Code Problem over a finite field $\F_q$,
denoted $\MinDist(q)$, asks for a non-zero codeword with minimum
Hamming weight in a given linear code $C$ (i.e., a linear subspace of
$\F_q^n$). The problem was proved to be NP-hard by Vardy \cite{Vardy}.

Dumer, Micciancio, and Sudan \cite{DMS} proved that assuming RP $\ne$
NP the problem is hard to approximate within some factor $\gamma > 1$
using a {\it gap preserving} reduction from the Nearest Codeword
Problem, denoted $\NCP(q)$ (which is known to be NP-hard even with a
large gap). The latter problem asks, given a code $\tilde{C} \subseteq
\F_q^{m}$ and a point $p \in \F_q^m$, for a codeword that is nearest
to $p$ in Hamming distance.  However, Dumer \etal's reduction is
randomized: it maps an instance $(\tilde{C}, p)$ of $\NCP(q)$ to an
instance $C$ of $\MinDist(q)$ in a randomized manner such that: in
the YES Case, with high probability, the code $C$ has a non-zero
codeword with weight at most $d$, and in the NO Case, $C$ has no
non-zero codeword of weight less that $\gamma d$, for some fixed
constant $\gamma > 1$. We note that the minimum distance of code is
multiplicative under the tensor product of codes; this enables one to
{\it boost} the inapproximability result to any constant factor, or
even to an {\it almost polynomial factor} (under a quasipolynomial
time reduction), see \cite{DMS}.

The randomness in Dumer \etal's reduction is used for constructing, as
a gadget, a non-trivial coding theoretic construction with certain
properties (see Section~\ref{sec:comparison} for details).  In a
remarkable pair of papers, Cheng and Wan \cite{CW1, CW2} recently
constructed such a gadget deterministically, thereby giving a
deterministic reduction to the Gap $\MinDist(q)$ Problem.  Cheng and
Wan's construction is quite sophisticated.  It is an interesting
pursuit, in our opinion, to seek an {\it elementary} deterministic
reduction for the Gap $\MinDist(q)$ Problem.

In this paper, we indeed present such a reduction. For codes over
$\F_2$, our reduction is (surprisingly) simple, and does not rely on
any specialized gadget construction.  The reduction can be extended to
codes over any finite field $\F_q$; however, then the details of the
reduction becomes more involved, and we need to use Viola's recent
construction of a psedorandom generator for low degree polynomials
\cite{Viola}.  Even in this case, the resulting reduction is
conceptuelly quite simple.

We also observe that our reduction produces asymptotically good codes,
i.e., having constant rate and relative distance.  While Dumer
\etal~\cite{DMS} are able to prove randomized hardness for such codes,
this was not obtained by the deterministic reduction by Cheng and Wan.
In \cite{CW2}, proving a constant factor hardness of approximation for
asymptotically good codes is mentioned as an open problem.

Our main theorem is thus:

\begin{theorem} 
  \label{thm:main}
  For any finite field $\F_q$, there exists a constant $\gamma > 0$
  such that it is NP-hard (via a deterministic reduction) to
  approximate the $\MinDist(q)$ problem to within a factor
  $1+\gamma$, even on codes with rate $\ge \gamma$ and relative
  distance $\ge \gamma$ (i.e., asymptotically good codes).
\end{theorem} 

As noted before, the hardness factor can be boosted via tensor product
of codes (though after a superconstant amount of tensoring the code is
no longer asymptotically good):

\begin{theorem} 
  For any finite field $\F_q$, and constant $\epsilon > 0$, it is hard
  to approximate the $\MinDist(q)$ problem to within a factor
  $2^{(\log n)^{1-\epsilon}}$ unless $\text{NP} \subseteq
  \text{DTIME}(2^{(\log n)^{O(1)}})$.
\end{theorem} 

Another motivation to seek a new deterministic reduction for
$\MinDist(q)$ is that it might lead to a deterministic reduction for
the analogous problem for integer lattices, namely the Shortest Vector
Problem (SVP). For SVP, we do not know of a deterministic reduction
that proves even the basic NP-hardness, let alone a hardness of
approximation result. All known reductions are randomized
\cite{Ajtai1, CN1, Mic1, Kh-svp, KhotSVP, HavivRegev}.  In fact, the
reduction of Dumer \etal~\cite{DMS} giving hardness of approximation
for $\MinDist(q)$ assuming NP $\ne$ RP is inspired by a reduction by
Micciancio~\cite{Mic1} for SVP.

Our hope is that our new reduction for $\MinDist(q)$ can be used to
shed new light on the hardness of SVP.  For instance, it might be
possible to combine our reductions for $\MinDist(q)$ for different
primes $q$ so as to give a reduction over integers, i.e., a reduction
to SVP.

\subsection{Previous Reductions}
\label{sec:comparison}

On a high level, the idea of the reduction of Dumer \etal~\cite{DMS}
is the following.  We start from the hardness of approximation for
$\NCP$.  Given an instance $(\tilde{C}, p)$, let us look at the code
$C = \text{span}(\tilde{C} \cup \{ p \})$.  Then any codeword of $C$
which uses the point $p$ must have large distance.  However, it can be
that the code $\tilde{C}$ itself has very small distance so that the
minimum distance of $C$ is unrelated to the distance from $p$ to
$\tilde{C}$.  Loosely speaking, the idea is then to combine $C$ with
an additional code $C'$ such that any codeword which does not use $p$
must have a large weight in $C'$.

Let us briefly describe the gadget of \cite{DMS}.  They use a coding
theoretic construction with the following properties (slightly
restated).  Let $\frac{1}{2} < \rho < 1$ be a fixed constant and $k$
be a growing integer parameter. The field size $q$ is thought of as a
fixed constant.

\begin{enumerate}
\item $C^* \subseteq \F_q^\ell$ ~is a linear code with distance $d$,
  where $\ell$ is polynomial in $k$ (think of $\ell = k^{100}$).
\item There is a ``center'' $v \in \F_q^\ell$ such that the ball of
  radius $r$ around $v$, denoted $B(v,r)$, contains $q^{k}$ codewords
  and $r = \lfloor \rho d \rfloor$. In notation, $| B(v,r) \cap C^* |
  \geq q^k$.
\item There is a linear map $T: \F_q^\ell \mapsto \F_q^{k'}$ such that
  the image of $B(v,r) \cap C^*$ under $T$ is the full space
  $\F_q^{k'}$. Here $k'$ is polynomial in $k$ (think of $k' =
  k^{0.1}$).
\end{enumerate}

Dumer \etal ~achieve such a construction in a randomized manner. They
let $C^*$ be a suitable concatenation of Reed-Solomon codes with the
Hadamard code so that even a {\it typical} ball of radius $r$ contains
many (i.e., $q^k$) codewords.  Hence choosing the center $v$ at random
satisfies the second property.  They show further that a random linear
map $T$ satisfies the third property.  By giving a deterministic
construction of such a gadget, Cheng and Wan \cite{CW1,CW2} recently
derandomized the reduction of \cite{DMS}.

\subsection{Organization}

We present a proof of this theorem for the binary field in Section
\ref{sec:2} and for a general finite field in Section
\ref{sec:q}. Even for the binary case, it is instructive to first see
a reduction to $\NCP(2)$ in Section \ref{sec:nc} which is then
extended to the $\MinDist(2)$ problem in Section \ref{sec:mindist2}.

\section{Preliminaries}

\subsection{Codes}

Let $q$ be a prime power.

\begin{definition} A linear code $C$ over a field $\F_q$ is a linear subspace
of $\F_q^n$, where $n$ is the block-length of the code and dimension of the
subspace $C$ is the dimension of the code. The distance of the code $d(C)$ is
the minimum Hamming weight of any non-zero vector in $C$.
\end{definition}

The two problems $\MinDist(q)$ and $\NCP(q)$ are defined as follows.

\begin{definition}
  $\MinDist(q)$ is the problem of determining the
  distance $d(C)$ of a linear code $C \subseteq \F_q^n$. The code may
  be given by the basis vectors for the subspace $C$ or
  by the linear forms defining the subspace.
\end{definition}

\begin{definition}\label{def:nc} 
  $\NCP(q)$ is the problem of determining the
  minimum distance from a given point $p \in \F_q^n$ to any codeword
  in a given code $C \subseteq \F_q^n$.  Equivalently, it is the
  problem of determining the minimum Hamming weight of any point $z$
  in a given affine subspace of $\F_q^n$ (which would be $C - p$).
\end{definition}

Our reduction uses tensor products of
codes, which are defined as follows.

\begin{definition}  Let $C_1, C_2 \subseteq \F_q^n$
  be linear codes.  Then the linear
  code $C_1 \tensor C_2 \subseteq \F_q^{n^2}$ is defined as the set of all
  $n \times n$ matrices over $\F_q$ such that each of its columns is a
  codeword in $C_1$ and each of its rows is a codeword in $C_2$.
\end{definition}

A well-known fact is that the distance of a code is multiplicative under the
tensor product of codes.

\begin{fact}
  \label{fact:tensordistance}
  Let $C_1, C_2 \subseteq \F_q^n$ be linear codes.  Then the linear
  code $C_1 \tensor C_2 \subseteq \F_q^{n^2}$ has distance $d(C_1 \tensor
  C_2) = d(C_1) d(C_2)$.
\end{fact}

We shall need the following Lemma which shows that for many codewords
of $C \tensor C$ one can obtain a stronger bound on the distance than
the bound $d(C)^2$ given by Fact~\ref{fact:tensordistance}.

\begin{lemma}
  \label{lemma:zerodiag}
  Let $C \subseteq \F_q^n$ be a linear code of distance $d = d(C)$, and let
  $Y \in C \tensor C$ be a non-zero
  codeword with the additional properties that
  \begin{enumerate}
  \item The diagonal of $Y$ is zero.
  \item $Y$ is symmetric.
  \end{enumerate}
  Then $Y$ has at least  $d^2(1+1/q)$ non-zero entries.
\end{lemma}


\begin{proof}
  Suppose $Y_{ij} = Y_{ji} \ne 0$.  Since we have $Y_{ii} = 0$ it must hold
  that $i \ne j$ and that rows $i$ and $j$ are linearly independent
  codewords of $C$.  By Fact~\ref{fact:distboost} below it follows
  that the number of columns $k$ such that at least one of $Y_{ik}$
  and $Y_{jk}$ is non-zero is at least $d(1+1/q)$.  Each of
  these columns then has  at least $d$ non-zero entries
  and hence $Y$ has at least  $d^2(1+1/q)$
  non-zero entries.
\end{proof}

\begin{fact}
  \label{fact:distboost}
  Let $C \subseteq \F_q^n$ be a linear code of distance $d = d(C)$.  Then for
  any two linearly independent codewords $x, y \in \F_q^n$, the number
  of coordinates $i \in [n]$ for which either $x_i \ne 0$ or $y_i \ne
  0$ is at least $d(1+1/q)$.
\end{fact}

\begin{proof}  Let $m$ be the number of coordinates such that
  $x_i \ne 0$ or $y_i \ne 0$ but not both, and let $m'$ be the
  number of coordinates such that both
  $x_i \ne 0$ and $y_i \ne 0$. Clearly,
   $$m + 2m' \geq 2d. $$ \
  We
  can choose $\lambda \not= 0$ appropriately so that the
vector $x - \lambda y$ has at most $m + m' - m'/(q-1)$ non-zero
entries. This implies
 $$m + m' - m'/(q-1) \geq d. $$
 Multiplying the first inequality by $1/q$, the second by
$(q-1)/q$, and adding up gives $m + m' \geq d(1+1/q)$ as desired.
\end{proof}

\subsection{Hardness of Constraint Satisfaction}

The starting point in our reduction is a constraint satisfaction
problem that we refer to as the $\MaxNAND$ problem, defined as
follows.

\begin{definition}
  An instance $\Psi$ of the $\MaxNAND{}$ problem consists of a
  set of quadratic equations over $\F_2$, each of the form $x_k =
  \text{NAND}(x_i, x_j) = 1 + x_i \cdot x_j$ for some variables $x_i, x_j,
  x_k$.  The objective is to find an assignment to the variables such
  that as many equations as possible are satisfied.  We denote by
  $\Opt(\Psi) \in [0,1]$ the maximum fraction of satisfied equations
  over all possible assignments to the variables.
\end{definition}

The following is an easy consequence of the PCP Theorem \cite{FGLSS, AS, ALMSS}
and
the fact that NAND gates form a basis for the space of boolean
functions.

\begin{theorem}
  There is a universal constant $\delta > 0$ such that given a
  $\MaxNAND{}$ instance $\Psi$ it is NP-hard to determine whether
  $\Opt(\Psi) = 1$ or $\Opt(\Psi) \le 1-\delta$.
\end{theorem}

\section{The Binary Case} \label{sec:2} 

In this section we give a  simple reduction from $\MaxNAND$
showing that it is NP-hard to approximate $\MinDist(2)$ to within
some constant factor.

\renewcommand{\S}{\mathcal{S}}

\subsection{Reduction to Nearest Codeword} \label{sec:nc} 

It is instructive to start with a reduction for the Nearest Codeword
Problem, $\NCP(2)$, for which it is significantly easier to prove
hardness.  There are even simpler reductions known than the one we
give here, but as we shall see in the next section this reduction can
be modified to give hardness for the $\MinDist(2)$ problem.

\medskip
Given a $\MaxNAND$ instance $\Psi$ with $n$ variables and $m$ constraints,
we shall construct an affine
subspace $\S$ of $\F_2^{4m}$ such that:
\begin{enumerate}
\item[(i)] If $\Psi$ is satisfiable then $\S$ has a vector of Hamming weight at most
 $m$.
\item[(ii)] If $\Opt(\Psi) \leq 1-2\delta$ then $\S$ has no vector of Hamming weight less than
$(1+2\delta)m$.
\end{enumerate}
This proves, according to Definition \ref{def:nc}, that $\NCP(2)$ is
NP-hard to approximate within a factor $1+2\delta$.

Every constraint $x_k = 1+x_ix_j$ in $\Psi$ gives rise to four new
variables, as follows.  We think of the four variables as a function
$S_{ijk}: \F_2^2 \rightarrow \F_2$.  The intent is that this function
should be the indicator function of the values of $x_i$ and $x_j$, in
other words, that
$$
S_{ijk}(a,b) = \left\{\begin{array}{ll}
1 & \text{if $x_i = a$ and $x_j = b$} \\
0 & \text{otherwise}
\end{array}\right..
$$
With this interpretation in mind, each function $S_{ijk}$ has to
satisfy the following linear constraints over $\F_2$:
\begin{eqnarray}
  S_{ijk}(0,0) + S_{ijk}(0,1) + S_{ijk}(1,0) + S_{ijk}(1,1) &=& 1
  \label{eq:sijk-1} \\
  S_{ijk}(1,0) + S_{ijk}(1,1) &=& x_i \label{eq:sijk-2} \\
  S_{ijk}(0,1) + S_{ijk}(1,1) &=& x_j \label{eq:sijk-3} \\
  S_{ijk}(0,0) + S_{ijk}(0,1) + S_{ijk}(1,0) &=& x_k. \label{eq:sijk-4}
\end{eqnarray}
Thus, we have a set of $n+4m$ variables $z_1, \ldots, z_{n+4m}$ (recall that
$n$ and $m$ are the number of variables and constraints of $\Psi$,
respectively) and $4m$ linear constraints of the form $\sum l_{ij} z_j
= b_i$ where $l_i \in \F_2^{n+4m}$ and $b_i \in \F_2$.

Let $\S \subseteq \F_2^{4m}$ be the affine subspace of $\F_2^{4m}$
defined by the set of solutions to the system of equations, projected
to the $4m$ coordinates corresponding to the $S_{ijk}$ variables.  Note
that these coordinates uniquely determine the remaining $n$
coordinates (assuming without loss of generality that every variable
of $\Psi$ appears in some constraint), according to Equations
\eqref{eq:sijk-2}-\eqref{eq:sijk-4}.

Now, if $\Psi$ is satisfiable, then using the satisfying assignment
for $x$ and the intended values for the $S_{ijk}$'s we obtain an
element of $\S$ with $m$ non-zero entries. Note that for each
constraint involving variables $x_i, x_j, x_k$, exactly one of the
four variables $S_{ijk}(\cdot, \cdot)$ is non-zero.

On the other hand, note that if the function $S_{ijk}(\cdot, \cdot)$
has exactly one non-zero entry it must be that the induced values of
$(x_i, x_j, x_k)$ satisfy the constraint $x_k = 1+x_i \cdot x_j$
(which one can see either by trying all such $S_{ijk}$ or noting that
each of the four different satisfying assignments to $(x_i, x_j, x_k)$
gives a unique such $S_{ijk}$).  Since every $S_{ijk}$ is constrained
to have an odd number of non-zero entries by Equation \eqref{eq:sijk-1},
it means that whenever
$S_{ijk}$ induces values of $(x_i, x_j, x_k)$ that do not satisfy $x_k
= 1+x_i \cdot x_j$, it must hold that $S_{ijk}$ has three non-zero
entries.  Therefore, we see that if $\Opt(\Psi) \le 1-\delta$, it must
hold that every element of $\S$ has at least $(1+2\delta)m$ non-zero
entries.

To summarize, we obtain that it is NP-hard to approximate the minimum
weight element of an affine subspace (or equivalently, the Nearest
Codeword Problem) to within a constant factor $1 + 2\delta$.

\subsection{Reduction to Minimum Distance} \label{sec:mindist2} 

To get the hardness result for the $\MinDist$ problem, we would like
to alter the reduction in the previous section so that it produces a
linear subspace rather than an affine one.  The only non-homogenous
part of the subspace produced are the equations \eqref{eq:sijk-1}
constraining each $S_{ijk}$ to have an odd number of entries.  To
produce a linear subspace, we are going to replace the constant $1$
with a variable $x_0$, which is intended to take the value $1$.  In other words, we replace Equation~\eqref{eq:sijk-1} with the following equation:
\begin{equation}
\tag{\ref*{eq:sijk-1}'}
\label{eq:sijk-1x}
S_{ijk}(0,0) + S_{ijk}(0,1) + S_{ijk}(1,0) + S_{ijk}(1,1) = x_0
\end{equation}
However, in order to make this work we need to
ensure that every assignment where $x_0$ is set to $0$ has large
weight, and this requires adding some more components to the reduction.

A first observation is that the system of constraints relating $S_{ijk}$ to
$(x_0, x_i, x_j, x_k)$ is invertible.  Namely, we have Equations
\eqref{eq:sijk-1x}-\eqref{eq:sijk-4}, and inversely, that
\begin{align*}
  S_{ijk}(0,0) &= x_i + x_j + x_k &
  S_{ijk}(0,1) &= x_0 + x_j + x_k \\
  S_{ijk}(1,0) &= x_0 + x_i + x_k &
  S_{ijk}(1,1) &= x_0 + x_k.
\end{align*}
Second, if $x_0 = 0$ but at least one of $(x_i, x_j, x_k)$ is
non-zero, it must hold that $S_{ijk}$ has at least two non-zero
entries.  Thus, if it happens that for a large fraction (more than
$1/2$) of constraints at least one of $(x_i, x_j, x_k)$ is non-zero,
it must be the case that the total weight of the $S_{ijk}$'s is larger
than $m$.  But of course, we have no way to guarantee such a condition
on $(x_i, x_j, x_k)$.

However, we can construct what morally amounts to a separate dummy
instance of $\MaxNAND$ that has this property, and then let it use the
same $x_0$ variable as $\Psi$.
Towards this end, let $C \subseteq \F_2^N$ be a linear code of relative
 distance
$1/2-\epsilon$.  Here $\epsilon > 0$ will be chosen sufficiently small
and for reasons that will
become clear momentarily, the dimension of the code will be exactly
 $n$ so that
one can take $N = O(n)$.

\smallskip
Now we introduce $N + N^2$ new variables which we think of as a vector
$y \in \F_2^N$ and matrix $Y \in \F_2^{N \times N}$.  The vector $y$
should be an element of $C$ and the matrix $Y$ should be an element of
$C \tensor C$.  The intention is that $Y = y \cdot y^{\top}$, or in other
words, that for every $i, j \in [N]$ we have $Y_{ij} = y_i \cdot y_j$.

Analogously to the $S_{ijk}$ functions intended to check the NAND
constraints of $\Psi$, we now introduce for every $i,j \in [N]$ a
function $Z_{ij}: \F_2^2 \rightarrow \F_2$ that is intended to check
the constraint $Y_{ij} = y_i \cdot y_j$, and that is supposed to be
the indicator of the assignment to the variables $(y_i, y_j)$.  We
then impose the analogues of the constraints
\eqref{eq:sijk-1x}-\eqref{eq:sijk-4},  viz.
\begin{eqnarray}
  Z_{ij}(0,0) + Z_{ij}(0,1) + Z_{ij}(1,0) + Z_{ij}(1,1) &=& x_0
  \label{eq:zij-1} \\
  Z_{ij}(1,0) + Z_{ij}(1,1) &=& y_i \label{eq:zij-2} \\
  Z_{ij}(0,1) + Z_{ij}(1,1) &=& y_j \label{eq:zij-3} \\
  Z_{ij}(1,1) &=& Y_{ij}.\label{eq:zij-4}
\end{eqnarray}
Figure~\ref{fig:f2reduction} gives an overview of the different
components of the reduction and their relations (including some
relations that we have not yet described, though we shall do so
momentarily).

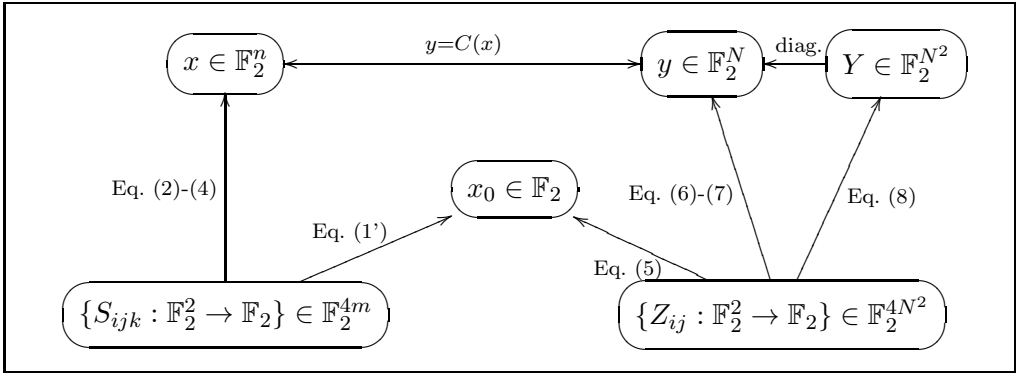
\begin{figure}[ht]
  \framebox{
    \begin{minipage}{0.85\textwidth}
      \entrymodifiers={++[F-:<10pt>]}
      \vspace{0.2cm}
      \centerline{
        \xymatrix{
          {x \in \F_2^n} \ar@{<->}[rr]^{y = C(x)} & 
          *{} & 
          {y \in \F_2^N} &
          {Y \in \F_2^{N^2}} \ar[l]_{\text{diag.}} \\
          *{} & 
          {x_0 \in \F_2} & 
          *{} \\
          \{S_{ijk}: \F_2^2 \rightarrow \F_2\} \in \F_2^{4m} \ar[uu]^{\text{Eq.~\eqref{eq:sijk-2}-\eqref{eq:sijk-4}}} \ar[ur]^{\text{Eq.~\eqref{eq:sijk-1x}}} &
          *{} & 
          *+<26pt>{} & *{}
          \save [l].[]!C="Z"*++[F-:<10pt>]{\{Z_{ij}: \F_2^2 \rightarrow \F_2\} \in \F_2^{4N^2}}\frm{}\restore
          \ar "Z";"2,2"^{\text{Eq.~\eqref{eq:zij-1}}}
          \ar "Z";"1,3"^{\text{Eq.~\eqref{eq:zij-2}-\eqref{eq:zij-3}}}
          \ar "Z";"1,4"_{\text{Eq.~\eqref{eq:zij-4}}}
        }
      }
      \vspace{0.2cm}
    \end{minipage}
  }
  \centering
  \caption{The different components of the reduction to $\MinDist(2)$.
    An arrow from one component to another indicates that the second
    component is a linear function of the first, with the label
    indicating the nature of this linear function.}
  \label{fig:f2reduction}
\end{figure}

The final subspace $\S$ will consist of the projection to the $4m$
different $S_{ijk}$ variables and the $4N^2$ different $Z_{ij}$
variables, but with each of the $S_{ijk}$ variables repeated some $r
\approx N^2/m$ number of times in order to make these two sets of
variables of comparable size.

Note that by Equations~\eqref{eq:sijk-1}-\eqref{eq:sijk-4} and
\eqref{eq:zij-1}-\eqref{eq:zij-4} these variables uniquely determine
$x_0$, $x$, $y$ and $Y$.  Furthermore, because of the invertibility of
these constraints, we have that if some $S_{ijk}$ or $Z_{ij}$ is
non-zero it must hold that one of $x_0$, $x$, $y$ and $Y$ are
non-zero.

As in the previous section, when $x_0$ is non-zero, each $S_{ijk}$ and
$Z_{ij}$ must have at least one non-zero entry and all the $\delta$
fraction of the $S_{ijk}$'s corresponding to unsatisfied NAND
constraints of $\Psi$ must have at least three non-zero entries, giving a total
weight of
$$
(1+2\delta)rm + N^2.
$$
Now consider the case that $x_0$ is zero. Let us first look at the
subcase that $y$ is non-zero. Since
$y \in C$ is a non-zero codeword, at least $(1/2-\epsilon)N$ of its
coordinates are non-zero. Thus, for
at least $(3/4 - 2\epsilon)N^2$  pairs $(y_i, y_j) \not= (0,0)$. For each
such pair, the corresponding $Z_{ij}$ function
is non-zero, and as argued earlier,
has at least two non-zero
entries, which means that the total weight of the $Z_{ij}$'s is at least
$$2 \cdot (3/4-2\epsilon) \cdot N^2 = \left(\frac{3}{2}- 4\epsilon\right) \cdot N^2.$$

The next subcase is that $x_0$ and $y$ are zero but either $x$ or $Y$
is non-zero.  We first enforce that $x=0$.   Recall that  $C$ has
dimension exactly $n$, and hence there is a one-to-one linear map
$C : \F_2^n \mapsto \F_2^N$. We may  therefore
 add the additional constraints that $y = C(x)$ is the
encoding of $x$.  Then,  $x$ is non-zero if and only if $y$ is.

The only possibility that remains is that $x_0$, $x$ and $y$ are all
zero, but that the matrix $Y$ is non-zero.  In this case, it is easily
verified from Equations \eqref{eq:zij-1}-\eqref{eq:zij-4}
that for each $i,j \in [N]$ such that $Y_{ij}$ is non-zero,
it must be that $Z_{ij}$ has four non-zero entries.  However, the
distance of the code $C \tensor C$ to which $Y$ belongs is only
$(1/2-\epsilon)^2 < 1/4$, so it seems as though we just came short of
obtaining a large distance.  However, there are two additional
constraints that we can impose on $Y$: first, if $Y = y \cdot
y^{\top}$ we have that the diagonal entries $Y_{ii}$ should equal
$y_i^2 = y_i$, so we can add the requirement that the diagonal of $Y$
equals $y$.  Second, it should be the case that $Y_{ij} =
Y_{ji}$, so we also add the constraint that $Y$ is symmetric.  With
these constraints, Lemma~\ref{lemma:zerodiag} now implies that $Y$ in
fact has $(1/2-\epsilon)^2 \cdot \frac{3}{2} > (1/4-2\epsilon) \frac{3}{2}$ fraction
non-zero entries.  As mentioned above, each corresponding $Z_{ij}$
function has four non-zero entries giving a total of
$$
4 \cdot (3/8 - 3\epsilon) \cdot N^2 = \left(\frac{3}{2}-12\epsilon\right) \cdot N^2
$$
non-zero entries.

In summary, this gives that when $\Opt(\Psi) \le 1-\delta$, every non-zero
vector in $\S$ must have weight
at least $$\min\left((1+2\delta)rm + N^2,  \,
\left(\frac{3}{2}-12\epsilon\right) \cdot N^2 \right),$$ whereas if $\Psi$ is satisfiable the minimum
distance is $rm + N^2$ (since exactly one entry is non-zero for each
$S_{ijk}$ and $Z_{ij}$).   Choosing $\epsilon > 0$ sufficiently small and
$$r \approx \frac{N^2}{2(1+2\delta)m}$$ we obtain that it is NP-hard
to approximate $\MinDist(2)$ to within some factor $\delta' > 1$.  

We have not yet proved that $\C(\Psi)$ has good rate and distance.  In
Section~\ref{sect:goodcode}, we give a proof of this for our reduction
for the general case.  That proof also works for the binary case.

\section{Interlude: Polynomials and Pseudorandomness over $\F_q$}

In this section we describe some background material that we need for
the generalization of the reduction for $\F_2$ to any finite field.

We recall two basic properties about polynomials over finite fields.
First, we have the well-known fact that every function on $\F_q^n$
can be uniquely represented by a polynomial of maximum degree $q-1$.

\begin{fact}
  \label{fact:polybasis}
  The set of polynomials 
  $$\{\,X_1^{i_1}X_2^{i_2} \cdots X_n^{i_n}\,:\, \text{$0 \le i_j \le
    q-1$ for all $1 \le j \le n$}\,\}$$ form a basis for the set of
  functions from $\F_q^n$ to $\F_q$.
\end{fact}

Second, we have the Schwarz-Zippel Lemma.

\begin{lemma}[Schwarz-Zippel]
  \label{lemma:schwarzippel}
  Let $p \in \F_q [X_1, \ldots, X_n ]$ be a non-zero polynomial of
  total degree at most $d$.  Then $p$ has at most a fraction $d
  q^{n-1}$ zeros.
\end{lemma}

\subsection{Linear Approximations to Nonlinear Codes}

In our hardness result for $\MinDist(q)$, we need explicit
constructions of certain codes which can be thought of as serving as
linear approximations to some nonlinear codes.  In particular, we need
a sequence of linear codes $C_1, \ldots, C_{q-1}$ over $\F_q^N$ with
the following two properties:
\begin{enumerate}
\item $d(C_e) \gtrsim (1-e/q) \cdot N$ for $1 \leq e \leq q-1$.
\item If $x \in C_1$ then $x^e \in C_e$ for $1 \leq e \leq q-1$.
      Here $x^e$ denotes a vector that is componentwise $e^{th}$ power of $x$.
\end{enumerate}
In other words, $C_e$ should contain the nonlinear code $\{x^e\}_{x
  \in C_1}$, while still having a reasonable amount of distance.  In
this sense we can think of $C_e$ as a linear approximation to a
nonlinear code.  

To obtain such a sequence of codes, we use pseudorandom generators for
low-degree polynomials.  Such pseudorandom generators were recently
constructed by Viola \cite{Viola} (building on
\cite{BogVio10,Lovett09}), who showed that the sum of $d$ PRGs for
linear functions fool degree $d$ polynomials.  Using his result, and
PRGs against linear functions of optimal seed length $\log_q n + O(1 +
\log_q 1/\epsilon)$ (see e.g., Appendix A of \cite{BogVio10}), one
obtains the following theorem.

\begin{theorem}
  \label{thm:viola_generator}
  For every prime power $q$, $d > 0$, $\epsilon > 0$ there is a
  constant $c := c(q, d, \epsilon)$ such that for every $n > 0$, there
  is a polynomial time constructible (multi)set $R \subseteq \F_q^n$
  of size $|R| \le c \cdot n^{d}$ such that, for any polynomial $f:
  \F_q^n \rightarrow \F_q$ of total degree at most $d$, it holds that
  \begin{equation}
    \label{eqn:viola_generator1}
    \sum_{a \in \F_q} \left| \Pr_{x \sim R}[f(x) = a] - \Pr_{x \sim \F_q^n}[f(x) = a]\right| \le \epsilon.
  \end{equation}
\end{theorem}

\begin{remark}
  The constant $c$ of Theorem~\ref{thm:viola_generator}
  can be taken to be $c(q, d, \epsilon) = \left(q/\epsilon\right)^{O(d
    2^d)}$.
\end{remark}

\begin{remark}
  In order for the hardness result of Theorem~\ref{thm:main} to apply
  for codes with constant rate, we need the set $R$ of
  Theorem~\ref{thm:viola_generator} to have size $O(n^d)$.  For this,
  the parameters of Viola's result \cite{Viola} are necessary, and the
  earlier result \cite{Lovett09} does not suffice.  If one does not care
  about this property, any $|R| = \text{poly}(n)$ suffices.
\end{remark}

A simple corollary of the property (\ref{eqn:viola_generator1}) and
the Schwarz-Zippel Lemma~\ref{lemma:schwarzippel} is the following.

\begin{corollary}
  \label{cor:viola_generator}
  If $d = q-1$ the (multi)set $R \subseteq \F_q^n$ constructed in
  Theorem~\ref{thm:viola_generator} has the property that for
  every non-zero polynomial $f: \F_q^n \rightarrow \F_q$ of total
  degree at most $e \leq q-1$,
  \begin{equation}
    \label{eqn:violagenerator2}
    \Pr_{x \sim R}[f(x) \not= 0] \ge 1 - e/q - \epsilon.
  \end{equation}
\end{corollary}

Now define, for $1 \leq e \leq q-1$, $C_e$ to be the set of all
vectors $(f(x))_{x \in R}$ where $f: \F_q^n \mapsto \F_q$ is a degree
$e$ polynomial with no constant term (i.e.,  $f(0)=0$). Clearly, $C_e$
is a linear subspace of $\F_q^{|R|}$.  As observed in
Corollary~\ref{cor:viola_generator}, the relative distance of $C_e$ is
essentially $1-e/q$ (as $\epsilon $ can be taken to be arbitrarily
small relative to $q$). Moreover, any $v \in C_1$ is the evaluation
vector of a degree one polynomial, and hence $v^e$ is the evaluation
vector of a degree $e$ polynomial, and therefore $v^e \in C_e$ as
desired.

\section{Reduction to $\MinDist(q)$ for $q \geq 3$} \label{sec:q} 

We now describe a general reduction from the $\MaxNAND$ problem to the
$\MinDist(q)$ problem for any prime power $q$.  The basic idea is the
same as in the $\F_2$ case but some additional work is needed both in
the reduction itself and its analysis.

Given a $\MaxNAND$ instance $\Psi$, we construct a linear code
$\C(\Psi)$ over $\F_q$ as follows.  For simplicity we here assume that
$q \ge 3$ as the binary case was already handled in the previous
section.  As before, let $n$ be the number of variables in the $\MaxNAND$
 instance and $m$ the number of constraints.

Fix some small enough parameter $\epsilon$ and let $R \subseteq \F_q^n$ be
the $\epsilon$-pseudorandom set for degree $q-1$ polynomials $\F_q^n
\rightarrow \F_q$ given by Theorem~\ref{thm:viola_generator}.  Let $N
= |R| = O(n^{q-1})$.

For $0 \le d \le q-1$, let $P_d \subseteq \F_q^N$ be the linear
subspace of all degree $d$ polynomials in $n$ variables with
coefficients in $\F_q$ and no constant term, evaluated at points on
$R$. I.e., all vectors in $P_d$ are of the form $( p(x) )_{x \in R}$
for some polynomial $p \in \F_q [X_1, \ldots, X_n]$ with $\deg(p) \leq
d$ and $p(0)=0$.  Note that $P_d$ is a linear code and by
Corollary~\ref{cor:viola_generator}, its relative distance is at least
$1-d/q-\epsilon$.

We define $C = P_1$ and for $\alpha \in \F_q^n$ we write $C(\alpha)
\in \F_q^N$ for the encoding of $\alpha$ under $C$; this corresponds
to the evaluations of the linear polynomial $\sum_{i=1}^n \alpha_i
X_i$ at all points $(X_1, \ldots, X_n)$ in $R$.  Conversely, for a
codeword $y \in C$ we write $\alpha = C^{-1}(y) \in \F_q^n$ for the
(unique) decoding of $y$.

From here on, we will ignore the parameter $\epsilon > 0$; it can be
chosen to be sufficiently small (independent of $q$ and the
inapproximability for $\MaxNAND$) and hence the effect of this can be
made insignificant.


\medskip
We now construct a linear code $\C'(\Psi)$ with variables as described
in Figure~\ref{fig:vars}.  As in the $\F_2$ case, the final code
$\C(\Psi)$ will consist of the projection of these variables to the
$Z_{ij}$'s and the $S_{ijk}$'s, which determine the remaining
variables by the constraints that we shall define momentarily.


\begin{figure}[ht]
  \centering
  \framebox{
    \begin{minipage}{0.85\textwidth}
      \begin{enumerate}
      \item For every $0 \le e \le 2(q-1)$ a vector $Y^e \in \F_q^N$.
      \item For every $0 \le e,f \le q-1$ a matrix $Y^{e,f} \in \F_q^{N^2}$.
      \item For every $1 \le i, j \le N$ a function $Z_{ij}: \F_q^2
        \rightarrow \F_q$ (i.e., a vector in $\F_q^{q^2}$).
      \item For every equation $x_k = 1 + x_i \cdot x_j$ in $\Psi$, a function
        $S_{ijk}: \F_2^2 \rightarrow \F_q$ (i.e., a vector in $\F_q^4$).
      \end{enumerate}
    \end{minipage}
  }
  \caption{Variables of $\C'(\Psi)$.}
  \label{fig:vars}
\end{figure}

Before we describe the constraints defining $\C'(\Psi)$ it is
instructive to describe the intended values of these variables.
Loosely speaking, the different $Y$ variables are supposed to be an
encoding of an assignment $\alpha \in \F_2^n$ to $\Psi$, the function
$S_{ijk}$ is a check that $\alpha$ satisfies the equation $x_k = 1 +
x_i \cdot x_j$, and the $Z_{ij}$ functions check that the $Y$
variables resemble a valid encoding of some $\alpha$.

Specifically, the variables are supposed to be assigned as
described in Figure~\ref{fig:intent}.

\begin{figure}[ht]
  \centering
  \framebox{
    \begin{minipage}{0.85\textwidth}
      \begin{enumerate}
      \item $Y^e$ is supposed to be $C(\alpha)^e$ (where we think of
        $\F_2^n$ as a subset of $\F_q^n$ in the obvious way) .
      \item $Y^{ef}$ is supposed to be $C(\alpha)^e \cdot (C(\alpha)^f)^\top$ (i.e., we should have $Y^{ef}(i,j) = C(\alpha)_i^e C(\alpha)_j^f$.
      \item $Z_{ij}$ is supposed to be the indicator function of
        $(C(\alpha)_i,C(\alpha)_j)$ (i.e., $Z_{ij}(x,y)$ should be $1$ if $x
        = C(\alpha)_i$ and $y = C(\alpha)_j$; and $0$ otherwise).
      \item $S_{ijk}$ is supposed to be the indicator function of
        $(\alpha_i, \alpha_j)$ (i.e., $S_{ijk}(a,b) = 1$ if $\alpha_i = a$
        and $\alpha_j = b$; and $0$ otherwise).
      \end{enumerate}
    \end{minipage}
  }
  \caption{Intent of variables of $\C'(\Psi)$.}
  \label{fig:intent}
\end{figure}

We categorize the constraints of $\C'(\Psi)$ as being of two different
types, namely \emph{basic constraints} that aim to enforce rudimentary
checks of Items 1 and 2 of Figure~\ref{fig:intent}, and \emph{consistency
  constraints} that aim to use the $Z_{ij}$'s and $S_{ijk}$'s to check
that the $Y^{ef}$ matrices are consistent with an encoding of a good
assignment to $\Psi$.  As a comparison with the reduction for $\F_2$
in Section~\ref{sec:2}, the basic constraints correspond to the
horizontal arrows on the upper side of Figure~\ref{fig:f2reduction},
and the consistency constraints correspond to the other arrows, i.e.,
Equations \eqref{eq:sijk-1x}-\eqref{eq:zij-4}.

Keeping the interpretation from Figure~\ref{fig:intent} in mind, the
basic constraints that we impose are given in
Figure~\ref{fig:cons1}.

\begin{savenotes}  
\begin{figure}[ht]
  \centering
  \framebox{
    \begin{minipage}{0.85\textwidth}
      \begin{enumerate}
      \item For $0 \le e \le q-1$, $Y^e \in P_e$.
      \item For $q \le e \le 2(q-1)$, $Y^e = Y^{e-(q-1)}$.
      \item For $0 \le e,f \le q-1$:
        \begin{enumerate}
        \item $Y^{ef} \in P_e \tensor P_f$.
        \item The diagonal of $Y^{ef}$ equals $Y^{e+f}$.
        \end{enumerate}
      \item For $0 \le e \le q-1$, the rows (resp.\ columns) of $Y^{0,e}$
        (resp.\ $Y^{e,0}$) are identical (and therefore equal to $Y^e$ as
        this is the diagonal).
      \item The matrix $Y^{q-1,q-1}$ is symmetric\footnote{In general we could add the constraint that $Y^{e,f} = (Y^{f,e})^\top$ for every $e,f$, but it turns out we only need it for
        the case $e = f = q-1$.}.
      \end{enumerate}
    \end{minipage}
  }
  \caption{Basic constraints of $\C'(\Psi)$.}
  \label{fig:cons1}
\end{figure}
\end{savenotes}

Note that all entries of the matrix $Y^{0,0}$ must be equal, and that
in the intended assignment they should equal the constant $1$.  For
notational convenience let us write $Y_0 \in \F_q$ for the value of
the entries of $Y^{0,0}$ (this variable plays the same role as the
variable $x_0$ in the reduction for $\F_2$ in Section~\ref{sec:2}).

We then turn to the consistency constraints of $\C'(\Psi)$, which
are described in Figure~\ref{fig:cons2}.

\begin{figure}[ht]
  \centering
  \framebox{
    \begin{minipage}{0.85\textwidth}
      \begin{enumerate}
      \item
        For every constraint $x_k = 1 + x_i \cdot x_j$ of $\Psi$, four constraints
        on $S_{ijk}$:
        \begin{equation}
          \label{eq:sijk}
          \begin{aligned}
            Y_0 &= \sum_{a,b \in \F_2} S_{ijk}(a,b)  &
            \alpha_i &= \sum_{a,b \in \F_2} a \cdot S_{ijk}(a, b) \\
            \alpha_j &= \sum_{a,b \in \F_2} b \cdot S_{ijk}(a, b) &
            \alpha_k &= \sum_{a,b \in \F_2} (1 \oplus a \cdot b) \cdot S_{ijk}(a, b).
          \end{aligned}
        \end{equation}
        (Here $\oplus$ denotes addition in $\F_2$ and the remaining summations are over $\F_q$.)
      \item 
        For every $i, j \in [N]$, $q^2$ constraints on $Z_{ij}$: for
        every $0 \le e,f \le q-1$ it must hold that
        \begin{eqnarray}
          Y^{e,f}(i,j) &=& \sum_{x,y \in \F_q} x^e y^f Z_{ij}(x,y).  \label{eq:y-zij}
        \end{eqnarray}
      \end{enumerate}
    \end{minipage}
  }
  \caption{Consistency constraints of $\C'(\Psi)$.}
  \label{fig:cons2}
\end{figure}

The four equations \eqref{eq:sijk} are the same as Equations
\eqref{eq:sijk-1}-\eqref{eq:sijk-4} from the $\F_2$ reduction, the
only difference being that they are now constraints over $\F_q$.  Note
that instead of $Y_0$ we would like to use the constant $1$ in the
above constraint, but as we are not allowed to do this we use $Y_0$,
which, as mentioned above, is intended to equal $1$. Note also that
$Y^1 = C(\alpha)$, and thus $\alpha$ is implicitly defined by
$Y^1$. If one wanted to be precise, one would write $C^{-1}(Y^1)_i$
instead of $\alpha_i$ in the above equations.

Note that the function $S_{ijk}$ is an invertible linear
transformation of $\{Y_0, \alpha_i, \alpha_j, \alpha_k\}$ and hence is
non-zero if and only if one of those four variables are non-zero.
Similarly, from \eqref{eq:y-zij} it follows that $Z_{ij}$ is an
invertible linear transformation of the set of $(i,j)$'th entries of
the $q^2$ different matrices $\{Y^{ef}\}_{0 \le e,f \le q-1}$ (this is
an immediate consequence of Fact~\ref{fact:polybasis}).  In particular
$Z_{ij}$ is non-zero if and only if the $(i,j)$'th entry of some
matrix $Y^{e,f}$ is non-zero.

The final code $\C(\Psi)$ contains the projection of these variables
to the functions $Z_{ij}$ and the functions $S_{ijk}$, with each
$S_{ijk}$ repeated $r \ge 1$ times.  Note that $\C(\Psi)$ is a
subspace of $\F_q^M$ where $M = (qN)^2 + 4rm$. The completeness and soundness are as follows.

\begin{lemma}[Completeness]
  \label{lemma:completeness}
  If $\Opt(\Psi) = 1$ then $$d(\C(\Psi)) \le N^2 + rm.$$
\end{lemma}

\begin{lemma}[Soundness]
  \label{lemma:soundness}
  If $\Opt(\Psi) \le 1-\delta$ then $$d(\C(\Psi)) \ge \min\left(N^2 +
    (1+\delta)rm, (1+1/q)N^2 \right).$$
\end{lemma}

\begin{lemma}[$\C$ is a Good Code]
  \label{lemma:goodcode}
  The dimension of $\C(\Psi)$ is $\Omega(N^2)$, and the
  distance is at least $N^2$.
\end{lemma}

Setting $r \approx \frac{N^2}{(1+\delta)q m}$,
Lemmas~\ref{lemma:completeness}-\ref{lemma:goodcode} give
Theorem~\ref{thm:main} (for the case $q \ge 3$).

In the following three subsections we prove the three lemmas.

\subsection{Proof of Completeness}

We first consider the Completeness Lemma~\ref{lemma:completeness},
which is straightforward to prove.

\begin{proof}[Proof of Lemma~\ref{lemma:completeness}]
  Given a satisfying assignment $\alpha \in \F_2^n$ to the set of
  quadratic equations, we construct a good codeword by following the
  intent described in Figure~\ref{fig:intent}.  Clearly this satisfies
  all the basic constraints.

  To check the constraints on $Z_{ij}$, recall that it is defined as
  $$
  Z_{ij}(x,y) = \left\{ \begin{array}{ll}
      1 & \text{if $(x,y) = (C(\alpha)_i,C(\alpha)_j)$} \\
      0 & \text{otherwise}.
    \end{array}\right..
  $$
  This choice of $Z_{ij}$ satisfies its $q^2$ constraints since for any $0
  \le e,f \le q-1$
  $$
  \sum_{x,y} x^e y^f Z_{ij}(x,y) = C(\alpha)_i^e C(\alpha)_j^f = Y^{ef}(i,j).
  $$
  Analogously, for the constraints on $S_{ijk}$ we have
  $$
  S_{ijk}(a,b) = \left\{
  \begin{array}{ll}
    1 & \text{if $(a,b) = (\alpha_i,\alpha_j)$} \\
    0 & \text{otherwise}.
  \end{array}\right.,
  $$
  which is again easily verified to satisfy its four
  constraints and hence this constitutes a codeword.

  The weight of the codeword is $N^2 + rm$, since each $Z_{ij}$ and
  each $S_{ijk}$ has exactly one non-zero coordinate.
\end{proof}

\subsection{Proof of Soundness}

In this section we prove the Soundness Lemma~\ref{lemma:soundness},
which is the part that requires the most work.  Let us first describe
the intuition.

In the analysis, we view codewords where $Y_0 \ne 0$ as resembling a
valid encoding of some $\alpha \in \F_2^n$ and for these we shall
argue that small weight corresponds to a good assignment to $\Psi$.

Most of the complication comes from analysing codewords where $Y_0 =
0$, which we think of as not resembling a valid encoding of some
$\alpha$.  For such codewords we argue that there must be a lot of
weight on the $Z_{ij}$'s.  To pull off this argument, we look at a
non-zero $Y^{e,f}$ that has $d=e+f$ minimal.  Then we look at the set
of $Z_{ij}$'s that are non-zero.  The total number of such $Z_{ij}$'s
can be lower bounded using the distance bound on $Y^{e,f}$ (though
this bound unfortunately gets worse as $d$ increases).  The fact that
every $Y^{e',f'}$ with $e'+f' < d$ is zero gives a set of
$\Theta(d^2)$ linear constraints on every such $Z_{ij}$.  These
constraints induce a linear code over $\F_q^{q^2}$ to which each
$Z_{ij}$ must belong.  We then argue that as $d$ increases, the
distance of this linear code increases as well, meaning that the
non-zero $Z_{ij}$'s must have an increasingly larger number of
non-zero entries.  This increased distance balances the decrease in
the number of non-zero $Z_{ij}$'s, allowing us to conclude that no
matter the value of $d$, the total number of non-zero entries among
all the $Z_{ij}$'s is always large.

Before we proceed with the formal proof of the soundness, let us state
two lemmas that we use to obtain lower bounds on the distance of
$Z_{ij}$.  The proofs of these two lemmas can be found in
Section~\ref{sect:lemma_proofs}.  First, we have a lemma for the case
when $d$ is small.

\begin{lemma}
  \label{lemma:fdist_lowend}
  Suppose $f: \F_q \times \F_q \rightarrow \F_q$ is a non-zero
  function satisfying
  $$
    \sum_{x,y \in \F_q} x^ay^b f(x, y) = 0
    $$ for every $(a,b)$ such that $0 \le a, b \le q-1$ and $a+b < d$ for some
    $0 \le d \le q-1$.  Then $f(x,y) \ne 0$ for at least $d+1$ points in
    $\F_q^2$.
\end{lemma}

Second, we have a lemma for the case when $d$ is large.

\begin{lemma}
  \label{lemma:fdist_highend}
  Suppose $f: \F_q \times \F_q \rightarrow \F_q$ is a non-zero
  function satisfying
  $$
    \sum_{x,y \in \F_q} x^a y^b f(x, y) = 0
    $$
  for every $(a,b)$ such that $0 \leq a, b \leq q-1$ and $a+b < d$
  for some $q-1 \leq d \leq 2(q-1)$.   Then
  $f(x,y) \ne 0$ for at least $q(d+2-q)$ points in $\F_q^2$.
\end{lemma}

We are now ready to proceed with the proof of soundness.

\begin{proof}[Proof of Lemma~\ref{lemma:soundness}]

  Let $\{Z_{ij}\}_{i,j \in [N]}$ and $\{S_{ijk}\}_{(i,j,k) \in \Psi}$
  be some non-zero codeword of $\C(\Psi)$, and consider the induced
  values of the $Y$ variables.

  Let $(e,f)$ be such that $Y^{e,f}$ is non-zero and $e+f$ is minimal
  (breaking ties arbitrarily).  Since the codeword is non-zero it
  follows that such an $(e,f)$ exists (by invertibility of
  \eqref{eq:sijk} and \eqref{eq:y-zij}).

  We do a case analysis based on the value of $e+f$.

  \paragraph{Case 1:  $e = f = 0$.}

  This is the case when $Y_0 \ne 0$.  In other words, we think of the
  $Y$ variables as resembling a valid encoding of some assignment to
  $\Psi$, so that the soundness of $\Psi$ comes into play.

  If $e = f = 0$ we have that all $Z_{ij}$'s and $S_{ijk}$'s are
  non-zero and hence the weight is at least $N^2 + rm$.  We will show
  that the soundness condition of $\Psi$ implies that a $\delta$
  fraction of the $S_{ijk}$'s must in fact have two non-zero entries,
  so that the total weight of the codeword is at least
  $$
  N^2 + (1+\delta)rm.
  $$

  To see this, construct an assignment to the quadratic equations
  instance as follows. Let $\alpha = C^{-1}(Y) \in \F_q^n$.  From the
  $\alpha_i$, $i\in [n]$, we define a boolean assignment $\beta_i$ as
  follows:
  $\beta_i = 0$ if $\alpha_i = 0$, and $\beta_i =
  1$ otherwise.  We claim that every constraint $x_k = 1+x_i \cdot x_j$
  for which $S_{ijk}$ only has a single non-zero entry is satisfied by
  $\beta$.  Indeed, suppose that $S_{ijk}(a,b) = c \ne 0$ and all
  other values of $S_{ijk}$ are $0$.  Then the constraints on
  $S_{ijk}$ imply that
  \begin{align*}
    \alpha_i &= a \cdot c &
    \alpha_j &= b \cdot c &
    \alpha_k &= (1\oplus ab) \cdot c.
  \end{align*}
  which implies that $\beta_i = a$, $\beta_j = b$, and $\beta_k =
  1\oplus ab = 1 \oplus \beta_i \cdot \beta_j$.
  By the soundness assumption $\Opt(\Psi) \le 1-\delta$, and hence at least a
  $\delta$ fraction of the constraints are not satisfied by $\beta$;
   the corresponding $S_{ijk}$'s must therefore have at least two
  non-zero entries.

  \paragraph{Case 2: $0 < e+f < q-1$. }

  Let $d = e+f$. The minimality of $e+f$ implies that $Y^{a,b} \equiv 0$
  for all $a+b < d$. From
  Equation \eqref{eq:y-zij}, we have that for all $a+b < d$,
   $\sum_{x,y \in \F_q} x^a y^b Z_{ij}(x,y)=0$. Applying
   Lemma~\ref{lemma:fdist_lowend},  each non-zero
  $Z_{ij}$ has at least $d+1$ non-zero entries.  Furthermore the
  fraction of non-zero $Z_{ij}$'s is at least $1 - d/q$. This is
  because the distance of the codes $P_e$ and $P_f$
  is at least $1-e/q$ and $1-f/q$ respectively, and hence the distance of
  the code $P_e \tensor P_f$ is at least $(1-e/q)(1-f/q) \geq 1-d/q$. Thus
  at least a $1-d/q$ fraction of entries of $Y^{e,f}$ are non-zero and by
  Equation \eqref{eq:y-zij}, the same applies to $Z_{ij}$.
     Hence the
  total number of non-zero entries over all $Z_{ij}(\cdot,\cdot)$ is at least
  $$
  N^2 (1-d/q)(d+1) \ge N^2 \frac{2(q-1)}{q} \ge \frac{4}{3} N^2,
  $$
  where the first inequality follows by noting that for $1 \le d \le
  q-2$ the left hand side is minimized by $d =
  1$ and $d = q-2$, and the second inequality follows from the
  assumption $q \ge 3$.

  \paragraph{Case 3: $e+f = q-1$. }

  In this case, either of Lemma~\ref{lemma:fdist_lowend} or
  Lemma~\ref{lemma:fdist_highend} gives that any non-zero $Z_{ij}$ has
  $q$ non-zero entries.

  The fraction of $Z_{ij}$'s that are non-zero is at least
  $(1-e/q)(1-f/q) = 1/q + ef/q^2$.  Unfortunately, if $ef = 0$ this
  bound is not good enough.  However, note that if $Y^{0,q-1}$ (or
  $Y^{q-1,0}$) is non-zero then so is $Y^{q-1}$ (by
  Figure~\ref{fig:cons1}, item 4) implying that $Y^{q-2,1}$ is
  non-zero (since by Figure~\ref{fig:cons1}, item 3(b), it has
  $Y^{q-1}$ as diagonal).  Hence we may assume without loss of
  generality that $ef \ge q-2$ so that at least a fraction $1/q +
  (q-2)/q^2 = 2(q-1)/q^2$ of the $Z_{ij}$'s are non-zero.

  Thus we see that the total weight of the codeword is at least
  $$
  N^2 \cdot \frac{2(q-1)}{q^2} \cdot q = N^2 \cdot \frac{2(q-1)}{q} \ge \frac{4}{3} N^2.
  $$

  \paragraph{Case 4: $q-1 < e+f < 2(q-1)$.}

  Let $e+f = q-1+s$ for $1 \le s < q-1$.  In this case,
  Lemma~\ref{lemma:fdist_highend} gives that any non-zero $Z_ij$ has $q
  \cdot (e+f + 2 - q) = q (s+1)$ non-zero entries.  The fraction of
  $Z_{ij}$'s that are non-zero is at least $(1-e/q)(1-f/q) = 1-(e+f)/q +
  ef/q^2$.  Furthermore, since $0 \le e,f \le q-1$ we must have that
  $\min(e,f) \ge s$ so that $e f \ge s (q-1)$.  Hence
  $$
  1-(e+f)/q + ef/q^2 \ge 1 - \frac{q-1+s}{q} + \frac{s(q-1)}{q^2} = \frac{q-s}{q^2}
  $$
  Thus, the total weight of all the $Z_{ij}$'s is lower bounded by
  $$
  N^2 \cdot \frac{q-s}{q^2} \cdot q(s+1) = N^2 \cdot \frac{(q-s)(s+1)}{q} \ge N^2 \cdot \frac{2(q-1)}{q} \ge \frac{4}{3} N^2.
  $$

  \paragraph{Case 5: $e+f = 2(q-1)$. }

  The only remaining case is when $e = f = q-1$.  Now
  Lemma~\ref{lemma:fdist_highend} gives that any non-zero $Z_{ij}$ has
  $q^2$ non-zero entries.  On the other hand, {'a} priori, the
  distance of $Y^{q-1,q-1}$ is as small as $1/q^2$, which seems
  problematic.  However, we still have some leeway: recall that the
  diagonal of $Y^{q-1,q-1}$ should equal $Y^{2(q-1)} = Y^{q-1}$ which
  also happens to be the diagonal of $Y^{q-1,0}$
  (Figure~\ref{fig:cons1}, items 3(b) and 2).  Since $Y^{q-1,0}$ is
  identically $0$ this means that the diagonal of $Y^{q-1,q-1}$ has to
  be zero.  By Lemma~\ref{lemma:zerodiag}, we can then conclude that
  at least a fraction $\frac{1}{q^2} \cdot (1+1/q)$ of the $Z_{ij}$'s
  are non-zero.  As each such $Z_{ij}$ has $q^2$ non-zero entries, we
  see that the total weight of the codeword is at least
  $$
  N^2 \cdot (1+1/q).
  $$
This concludes the proof of Lemma~\ref{lemma:soundness}.
\end{proof}

\subsection{Proof That The Code Is Good}
\label{sect:goodcode}

In this section we prove Lemma~\ref{lemma:goodcode}, that $\C(\Psi)$
is a good code.  After the soundness analysis, this becomes relatively
easy.  To get the bound on the rate of the code, we need the following
simple lower bound on the rate of a certain restricted tensor product
of a code.

\begin{claim}
  \label{claim:tensorrate}
  Let $C \subseteq \F_q^n$ be a linear code and $\tilde{C}$ be the linear
  subspace of $C \tensor C$ where every codeword is restricted to be
  symmetric.  Then $\dim(\tilde{C}) \ge \dim(C)^2/2$.
\end{claim}

\begin{proof}
  Let $G \in \F_q^{n \times k}$ be the generator matrix of $C$, where
  $k = \dim(C)$.  It is easy to check that the generator matrix of $C
  \tensor C$ is $G \tensor G \in \F_q^{n^2 \times k^2}$.  We think of
  $G \tensor G$ as mapping a $k \times k$ matrix $X$ to an $n \times
  n$ matrix $Y = (G \tensor G) X$ where
  $$
  Y_{i_1, i_2} = \sum_{j_1, j_2 \in [k]} g_{i_1,j_1} g_{i_2,j_2} X_{j_1,j_2}.
  $$
  It is easily verified that if $X$ is symmetric then so is $Y$, so
  the dimension of $\tilde{C}$ is at least the dimension of the space
  of symmetric $k \times k$ matrices over $\F_q$, which equals
  $\frac{k(k+1)}{2} \ge k^2/2$
\end{proof}

We can now prove that $\C(\Psi)$ is a good code.

\begin{proof}[Proof of Lemma~\ref{lemma:goodcode}]
  Let us first consider the distance of $\C(\Psi)$.  In
  Lemma~\ref{lemma:soundness}, it is shown that any codeword for which
  $Y_0 = 0$ has at least $N^2(1+1/q) \ge N^2$ non-zero entries.  On
  the other hand, if $Y_0 \ne 0$ each $Z_{ij}$ and $S_{ijk}$ must have
  at least one non-zero entry, for a total of $N^2 + rm \ge N^2$
  non-zero entries.

  It remains to prove that $\C(\Psi)$ has large dimension, which
  requires a little more work.  Let $\alpha \in \F_q^n$ and assign
  every matrix $Y^{e,f}$ except $Y^{q-1,q-1}$ according to the intent
  of Figure~\ref{fig:intent}.  I.e., for $(e,f) \ne (q-1,q-1)$ we set
  $Y^{ef}(i,j) = C(\alpha)_i^e C(\alpha)_j^f$.

  We shall show that there are still $q^{\Omega(N^2)}$ ways to choose
  $Y^{q-1,q-1}$ so that the resulting set of values satisfy the basic
  constraints of Figure~\ref{fig:cons1}.  Then, from the
  invertibility of Equations~\eqref{eq:sijk} and \eqref{eq:y-zij} of
  Figure~\ref{fig:cons2}, it follows that each of these
  $q^{\Omega(N^2)}$ ways to choose $Y^{q-1,q-1}$ extends to a unique
  codeword of $\C(\Psi)$.

  By Claim~\ref{claim:tensorrate}, the space of matrices
  $Y^{q-1,q-1}$ satisfying Items 3(a) and 5 of Figure~\ref{fig:cons1}
  has dimension at least $\dim(P_{q-1})^2/2 \ge n^{2(q-1)}/2 =
  \Omega(N^2)$ (recall that $N = O(n^{q-1})$).  The only additional
  constraint on $Y^{q-1,q-1}$ is Item 3(b) of Figure~\ref{fig:cons1},
  that the diagonal has to be $Y^{2(q-1)} = Y^{q-1}$.  However, this
  can reduce the dimension by at most $N$, so the remaining dimension
  is still $\Omega(N^2)$.

\end{proof}

\section{Combinatorial Lemmas}
\label{sect:lemma_proofs}

In this section we prove the combinatorial lemmas used in the proof of
Lemma~\ref{lemma:soundness}.

\newtheorem*{lemma_fdist_lowend}{Lemma~\ref{lemma:fdist_lowend} restated}
\begin{lemma_fdist_lowend}
  Suppose $f: \F_q \times \F_q \rightarrow \F_q$ is a non-zero
  function satisfying
  $$
    \sum_{x,y \in \F_q} x^ay^b f(x, y) = 0
    $$ for every $(a,b)$ such that $0 \le a, b \le q-1$ and $a+b < d$ for some
    $0 \le d \le q-1$.  Then $f(x,y) \ne 0$ for at least $d+1$ points in
    $\F_q^2$.
\end{lemma_fdist_lowend}

\begin{proof}
  Let $X = \{\,x\,:\, \exists y \, f(x,y) \ne 0\,\}$ and $Y =
  \{\,y\,:\, \exists x\, f(x,y) \ne 0\,\}$.  Without loss of
  generality, assume that $|X| \ge |Y|$.
  Define $g: \F_q \rightarrow \F_q$ by
  $$
  g(x) = \sum_{y \in \F_q} f(x,y).
  $$
  First suppose $g$ is non-zero.  Then we use the fact that
  $$
  \sum_{x} x^a g(x) = \sum_{x,y} x^a y^0 f(x, y) = 0
  $$
  for every $a < d$, which implies that $g$ has to be non-zero in at
  least $d+1$ points. This is because in the $d \times q$ matrix whose
  rows are $(x^a)_{x \in \F_q}$ for $ 0 \leq a \leq d-1$, any $d$ columns
  form a Vandermonde matrix and hence are linearly independent. We used here
  the fact that $d \leq q-1$.
  Thus $f$ also has to be non-zero in $d+1$
  points and we are done.  Hence we can now assume that $g$ is
  identically $0$.

  Let $|X| = s$ and $|Y| = t$.  Since $g$ is identically $0$, it must
  hold that for any $x \in X$ there are at least two different $y$'s
  such that $f(x,y) \ne 0$, implying that $f$ is non-zero for at least
  $2s$ different points.  We now show that $s+t \ge d+2$ which implies
  that $s \ge \frac{d+2}{2}$ (since we assumed $s \ge t$) so that $f$
  must be non-zero on at least $d+2$ points.

  Consider the Vandermonde matrices
  \begin{align*}
  A_X &= \left( \begin{array}{ccccc}
      1  &x_1  & x_1^2 & \ldots & x_1^{s-1} \\
      1  &x_2  & x_2^2 & \ldots & x_2^{s-1} \\
      & \vdots&       &  \vdots&  \\
      1 & x_s & x_s^2 & \ldots & x_s^{s-1}
    \end{array}\right)    &
  A_Y &= \left( \begin{array}{ccccc}
      1  &y_1  & y_1^2 & \ldots & y_1^{t-1} \\
      1  &y_2  & y_2^2 & \ldots & y_2^{t-1} \\
      & \vdots&       &  \vdots&  \\
      1 & y_t & y_t^2 & \ldots & y_t^{t-1}
    \end{array}\right),
  \end{align*}
  where $x_1, \ldots, x_s$ are the elements of $X$ and $y_1, \ldots,
  y_t$ are the elements of $Y$. Since $A_X$ and $A_Y$ are
  non-singular, so is $B := (A_X \tensor A_Y)^{\top}$.  The matrix $B$
  is an $st \times st$ matrix such that
  for any $0 \le a < s$, $0 \le b < t$, its  $(a,b)$'th row is
  $(x_i^a y_j^b)_{i \in [s], j \in [t]}$. 

  Since $f$ is not identically zero on $X \times Y$ and $B$ is
  non-singular, the dot product of $f$ restricted to $X \times Y$ with
  some row of $B$ is non-zero, i.e., there exists a row $(a,b)$ such
  that
  $$ 0 \not= \sum_{i \in [s], j \in [t]} f(x_i,y_i) x_i^a y_j^b =
  \sum_{x,y \in \F_q} x^a y^b f(x,y),$$ where for the second equality
  we noted that $f$ is zero outside of $X \times Y$. From the
  hypothesis of the Lemma,
we must have $a + b \ge d$ and therefore
  $s+t \ge a+b+2 \ge d+2$.
\end{proof}

For the next lemma we first have the following easy claim.

\begin{claim}
  \label{claim:powsumvanish}
  Let $0 \le a \le q-2$. Then $\sum_{x \in \F_q} x^a = 0$.
\end{claim}

\begin{proof}
  The case $a = 0$ is trivial.  For $a > 0$, let $g$ be a generator
  for $\F_q$ and define $h = g^a$.  Since $1 \le a \le q-2$ we have $h \ne 1$
  and by Fermat's little theorem we have $h^{q-1} = 1$.  Thus we have
  $$
  \sum_{x \in \F_q} x^a = \sum_{i=0}^{q-2} (g^i)^a = \sum_{i=0}^{q-2}
  h^i = \frac{h^{q-1}-1}{h-1} = 0.
  $$
\end{proof}

Now we prove the second lemma.

\newtheorem*{lemma_fdist_highend}{Lemma~\ref{lemma:fdist_highend} restated}
\begin{lemma_fdist_highend}
  Suppose $f: \F_q \times \F_q \rightarrow \F_q$ is a non-zero
  function satisfying
  \begin{equation}
  \label{eqn:lowdegvanish}
    \sum_{x,y \in \F_q} x^a y^b f(x, y) = 0
  \end{equation}
  for every $(a,b)$ such that $0 \leq a, b \leq q-1$ and $a+b < d$
  for some $q-1 \leq d \leq 2(q-1)$.   Then
  $f(x,y) \ne 0$ for at least $q(d+2-q)$ points in $\F_q^2$.
\end{lemma_fdist_highend}

\begin{proof}
  Let
  \begin{eqnarray*}
    S &=& \{\,(a,b)\,:\ 0 \leq a,b \leq q-1, \ a+b < d\,\} \\
    T &=& \{\,(e,\ell)\,:\ 0 \leq e, \ell \leq q-1, \ e+\ell \le 2(q-1)-d\,\}
  \end{eqnarray*}
  Note that $|S| + |T| = q^2$ since the mapping $(e,\ell) \mapsto
  (q-1-e,q-1-\ell)$ forms a bijection from $T$ to $\{0,1,\ldots,q-1\}^2 \setminus S$.

  Now, the functions $f$ satisfying \eqref{eqn:lowdegvanish} for every
  $(a,b) \in S$ form a linear subspace $V$ of $\F_q^{q^2}$ of
  dimension $q^2-|S| = |T|$.

  We identify the following basis for $V$: for every $(e,\ell) \in T$,
  let $g_{e\ell}(x,y) = x^e y^\ell$.  It is clear that the
  $g_{e\ell}$'s are linearly independent (since they are a subset of
  the standard polynomial basis for functions $\F_q^2 \rightarrow
  \F_q$; Fact~\ref{fact:polybasis}) and that $|\{g_{e\ell}\}| = |T| =
  \dim V$, so we only have to check that each $g_{e\ell}$ indeed lies
  in $V$.  We have
  $$
  \sum_{x,y} x^ay^b g_{e\ell}(x,y) = \sum_{x,y} x^{a+e}y^{b+\ell} = \left( \sum_{x} x^{a+e} \right) \cdot \left(\sum_{y} y^{b+\ell} \right).
  $$
  By Claim~\ref{claim:powsumvanish}, we see that this vanishes if
  either $a+e < q-1$ or $b+\ell < q-1$.  But this must hold, since
  otherwise we would have $(a+b)+(e+\ell) \ge 2(q-1)$ contradicting that
  $(a,b) \in S$ and $(e,\ell) \in T$.

  From this we can conclude that any function $f: \F_q^2 \rightarrow
  \F_q$ satisfying condition (\ref{eqn:lowdegvanish}) can be written
  as a polynomial of total degree at most $2(q-1)-d$.  By the
  Schwarz-Zippel Lemma \ref{lemma:schwarzippel} a non-zero such $f$ can be zero on at most
  a fraction $\frac{2(q-1)-d}{q}$ points of $\F_q^2$ and so $f$ has to
  be non-zero on at least
  $$q^2\left(1 - \frac{2(q-1)-d}{q}\right) = q(d+2-q)$$
  points.
\end{proof}

\bibliographystyle{abbrv}
\bibliography{references}

\end{document}